\documentclass[article, aps prl,superscriptaddress,longbibliography,groupedaddress]{revtex4}  % for review and submission
\usepackage{graphicx}  % needed for figures
\usepackage{dcolumn}   % needed for some tables
\usepackage{bm}        % for math
\usepackage{amssymb}   % for math
\usepackage{subfigure}
\usepackage{amsthm}
\usepackage{amsmath}

\def\setof#1{\left\{{#1}\right\}}

\newcommand{\R}{{\mathbb{R}}}
\newcommand{\N}{{\mathbb{N}}}

\newcommand{\CDF}{{\mbox{CDF}}}

\newtheorem{theorem}{Theorem}
\newtheorem{lemma}{Lemma}

\begin{document}

% The following information is for internal review, please remove them for submission

\title{Quantitative Measure of Memory Loss in Complex Spatio-Temporal Systems  }

\author{ Miroslav Kram\'ar}
\affiliation{
INRIA Saclay, 1 Rue Honor d'Estienne d'Orves, 91120 Palaiseau, France
}
\author{Lenka Kovalcinova}
\affiliation{Department of Mathematical Sciences,
New Jersey Institute of Technology,
University Heights,
Newark, NJ 07102}
\author{Konstantin Mischaikow}
\affiliation{Department of Mathematics and BioMaPS Institute,
Hill Center-Busch Campus,
Rutgers University,
110 Frelinghusen Rd,
Piscataway, NJ  08854-8019, USA}
\author{Lou Kondic}
\affiliation{Department of Mathematical Sciences,
New Jersey Institute of Technology,
University Heights,
Newark, NJ 07102}

\date{\today}

\begin{abstract}
To make progress in understanding the issue of memory loss and history dependence in  evolving complex systems, we consider the mixing rate that specifies how fast the future states become independent of the initial condition.    We propose a simple measure for assessing the mixing rate  that can be directly applied to  experimental data observed in any metric space $X$. For a compact phase space $X \subset R^M$,  we prove the following statement.  If  the underlying dynamical system has a unique physical measure and its dynamics is strongly mixing with respect to this measure, then our method provides an upper bound of the mixing rate.   We employ our method to analyze memory loss for the system of slowly sheared granular particles with a small inertial number $I$. The shear is induced by the moving  walls as well as by the linear motion of the support surface that ensures approximately linear shear throughout the sample.   We show that even if $I$ is kept fixed, the rate of memory loss (considered at the time scale given by the inverse shear rate)  depends erratically on the shear rate. Our study suggests a presence of bifurcations at which the rate of memory loss increases with the shear rate while it decreases away from these points.   We also find that the memory loss is not a smooth process. Its rate is closely related to frequency of the sudden transitions of the force network.   The loss of memory, quantified by observing evolution of force networks, is found to be   correlated  with the loss of correlation of shear stress measured on the system scale. Thus, we have  established a direct link between the evolution of force networks and macroscopic properties of the considered system.  
\end{abstract}

%\pacs{}
\maketitle

\section{Introduction.}
Understanding the global dynamics of nonlinear systems is typically challenging, especially if the governing equations are not known and only experimental data are available. 
A significant challenge arises from chaotic behavior as this makes precise long term prediction of future states impossible.
While low dimensional chaotic dynamics is relatively well understood via the geometric theory of differential equations \cite{guckenheimer2013nonlinear, robinson1998dynamical}, higher dimensional systems are usually treated using the ergodic theory \cite{walters2000introduction, petersen1989ergodic}.

In this Letter we propose a quantitative measure to assess the rate at which an ergodic system looses its memory, i.e.\ how fast the future states become independent of the initial condition.  The main idea is to analyze distributions of distances between  consecutive states of the system sampled at different sampling rates.   We stress that knowledge of the governing equations of the system is not required and the method can be applied to experimental data. However, to obtain a rigorous upper bound of the mixing rate for the underlying dynamical system $f\colon X\to X$ we have to require that: (i)  $X \subset R^M$ is compact, (ii) the system has a unique physical invariant measure $\mu$ and is strongly mixing with respect to this measure, i.e. $\mu(f^{-m}(A)\cap B)$ converges to $\mu(A)\mu(B)$ as $m\to \infty$ for any two measurable sets $A,B\subset X$.  There exist a few slightly different definitions of  the mixing rate. We consider it to be the rate at which the strong mixing coefficient $\alpha$ converges to zero~ \protect \cite{bradley2005basic}.   Our result provides an upper bound on  convergence rate of $\alpha$  and we employ it to quantify the memory loss of the system. 

First we familiarize the reader with our method by considering discrete dynamical systems generated by the tent map and the logistic map. However,  our main motivation comes from the dynamics of a large number of interacting particles modeling granular systems. Granular systems have attracted a lot of attention over the last centuries because of their importance to everyday life. However, many of their properties remain obscure even today. In particular, any description of a granular system in terms of particles' positions and moments is necessarily incomplete, since it sheds limited light on the properties of the particle interactions; even for the simplest case of a static system of (frictional) granular particles, the interaction field is not uniquely determined by the particle positions. Both physical experiments and simulations have shown that the interaction field consists of complex interaction networks that are known to be crucial for understanding mechanical properties of the system~\cite{networks_review_18}.   Properties of these networks can be well described using tools of algebraic topology \cite{dijksman2018characterizing}.  In particular, persistence diagrams \cite{carlsson, edelsbrunner-harer} provide quantitative succinct  descriptions of the changes in topology of an interaction network as the force level changes from infinity to zero.  For a detailed treatment of the subject we refer the reader to~\cite{physicaD14} while a more compact presentation can be found in~\cite{pre13}.  The space of persistence diagrams is a metric space and one can study the dynamics  of interaction networks~\cite{pre14, pg17} in this space.  

%LK describe better the regime considered - dense flows, etc.
This Letter presents an application of our measure of the rate of memory loss  to a simple shear flow of
granular particles in the regime where inertial effects are strong, and particles are stiff. We show that dependence of the rate of memory loss on the parameters of the system is nonlinear and not at all obvious. 
Our study shows that the global dynamics of the system changes erratically with the control parameter suggesting a presence of bifurcations even if the inertial  number, $I$, ~\cite{cruz_2005} is kept fixed. 
We also show that the rate of memory loss is closely connected to the frequency of abrupt changes in the persistence diagrams of the interaction networks, triggered by sudden reorganization of the force network. Moreover, we show that these changes are correlated to system-wide average measures such as shear stress. We document this by demonstrating a relation between decay of autocorrelation of shear stress and our measure of the memory loss.

\begin{figure}
\includegraphics{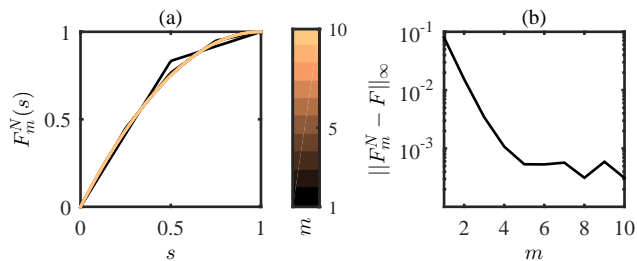}
\caption{(a) $\CDF$s $F^N_m$ of $\setof{d(x_n, x_{n+m})}_{n=0}^{N-m}$ for the orbit of the ten map starting at $x=0.1$ and $N = 5 \times 10^6$.  (b) Rapid decay of $|| F^N_m - F ||_\infty$ shows that   $F^N_m$ gets close to  $F$  as $m$ increases.} 
\label{fig::TentFunction}
\end{figure}

\section{ Mixing rate.} We propose that the mixing rate can be inferred from the convergence rate of the  cumulative distribution functions (CDFs),  $F_\tau$, of the distances  $d(x(t),x(t+\tau))$ along a  trajectory $x$. 
To provide motivation and intuition for our result we consider the discrete system generated by the tent map $f(x) := 2\min \setof{x, 1-x}$  on the interval $[0,1]$.  
The physical invariant measure $\mu$ of this system coincides with the Lebesgue measure on $[0,1]$. 
We begin by  computing  a finite sample $\setof{x_i}_{i = 0}^N$, $N = 5\times 10^6$,  of a trajectory starting from  $x_0 =0.1$, see the Section~\ref{sec:tent_map} for more details. 
From this sample, we calculate the CDFs, $F_m^N$, of the distances  $\setof{d(x_n, x_{n+m})}_0^{N-m}$ for different values of $m$.  
Figure~\ref{fig::TentFunction}(a) suggests that  $F_m^N$  converge as $m \to \infty$.  
In Section~\ref{sec:strongly_mixing_systems} we prove that for the CDFs $F_m$ of the random variable $d(x,f^m(x))$, with  $x$  distributed according to  $\mu$, there exits $C>0$  such that 
\begin{equation}
\label{eqn::Tent_convergence}
|| F_m - F ||_\infty := \sup_{s \in \R}| F_m(s) - F(s) | < C2^{-m},
\end{equation}
where $F$ is the CDF of  the distance $d(x,y)$ between two random variables  $x$ and $y$ which are i.i.d. according to $\mu$. Figure~\ref{fig::TentFunction}(b) shows that  $|| F^N_m - F ||_\infty$ exhibits the decay  predicted by  (\ref{eqn::Tent_convergence})  only for $m \leq 5$. For $m > 5$ a larger sample is needed to  properly approximate  'spatial averages'  $F_m$ by  'time averages' $F_m^N$.

It might seem that the decay given by~(\ref{eqn::Tent_convergence}) is closely connected to the Lyapunov exponent of the system. The tent map locally stretches the space by factor of two and its Lyapunov exponent is $\ln2$. However, the following theorem shows that the convergence rate reflects the mixing rate, which has a more global nature than the Lyapunov exponent. 

\begin{theorem}
\label{thm::convergence_rate}
 Let $X\subset \R^M$ be compact and suppose that $f\colon X\to X$ has a unique invariant measure, $\mu$, whose Radon-Nikodym derivative is continuous  with respect to the Lebesgue measure on $\R^M$.  Let $\setof{\varepsilon_m}_{m=1}^{\infty}$ be a sequence of positive numbers converging to zero. If there exits a sequence of partitions $\setof{\mathcal{T}_m}_{m=1}^{\infty}$ of $X$ such that diameter of every set $T \in \mathcal{T}_m$ is less than $\varepsilon_m$ and for every measurable set $E \subset X$
\begin{equation}
\label{eqn::mixing}
|\mu(E)\mu(T) - \mu(f^{-m}(E) \cap T) | <  \varepsilon_m \mu(T),
\end{equation}
then there exits a constant $C > 0$ such that
\begin{equation}
\label{eqn::convergence}
|| F_m - F ||_\infty < C\varepsilon_m.
\end{equation}
\end{theorem}

The relevant consequence of this theorem is that the convergence rate of $F_m$ provides an upper bound on how fast the sets with diameter $\varepsilon_m$ are mixed.  
As the system is mixed the trajectories starting from similar initial conditions become independent. Thus our ability to forecast the future states decreases, and  we can interpret the mixing rate as a rate of memory loss of the system. 
For further intuition the reader is referred to Section~\ref{sec:logistic_map} where a  study of the logistic equation is presented.

\section{ Continuous dynamics.} 
Our model for the granular dynamics takes the form of a differential equation $\dot{x}=h(x)$ where $x \in \R^M$.
For a short time scale $\tau$ the dynamics  tends to be well approximated by the linearization of the system and  $d(x(t), x(t + \tau)) \approx \tau ||h(x(t))||_2 $ where $|| \cdot ||_2$ is  Euclidian norm in $\R^M$.  
To detect the longest time scale for which the system is well approximated by its linearization, let us  consider the  CDFs  $F_m$ of the distance $d(x(t),x(t+m\tau))$ along an orbit $x$. If the dynamics of $\dot{x}=h(x)$ is well approximated by its linearization  at the time  scale $m\tau$, then $d(x(t),x(t+m\tau)) \approx  m d(x(t),x(t+\tau))$ and  $F_m(s) \approx F_1(ms)$. This scaling is lost at the time scale at which the non-linear effects become important. Moreover, if the functions $F_m$ are essentially constant for $m > m_0$, then it is likely that the system has a bounded attractor and its subsets with small volume are well mixed at the time scale $m_0\tau$.

\begin{figure}
\includegraphics{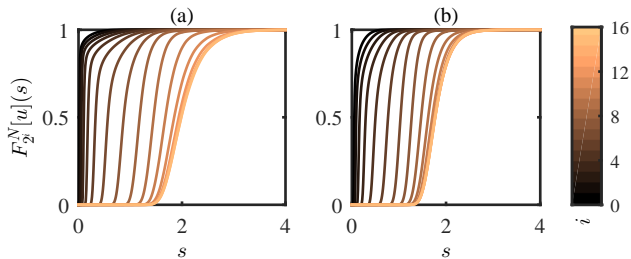}
\caption{  $F^N_{2^i}[u]$ for the system (a) $S_1$ and (b) $S_5$. The value of $i$ is indicated by the color bar. }
\label{fig::Dist}
\end{figure}

\section{ Linearly sheared systems.}  Now we consider simulations of sheared granular particles.   The details of the 
simulation protocol are given in Section~\ref{sec:simulations}; here we provide an overview.  
In order to avoid as much as possible the complications involved in inhomogeneous flows, but also governed by the goal to 
consider experimentally realizable configurations, we consider the following setup. 
Two-dimensional frictional bidisperse circular particles (elastic disks) are placed between solid walls that impose shear flow (shear
rate $\dot \gamma$) by moving to the right (top) and
left (bottom) with the same speed, $v$, (all the relevant quantities in what follows are 
non-dimensionalized using average particle diameter, mass, and binary
collision time, $\tau_c$, as the length, mass, and time scale, respectively, see Section~\ref{sec:simulations}.
To ensure uniform flow, the particles are placed on a solid substrate that moves with 
a linear velocity profile, similarly as in recent experiments~\cite{dapeng}.   
The walls are subject to the applied pressure, $P$, chosen in such a way that the interaction forces between the particles are orders
of magnitude stronger than the particle/substrate forces.   In the present work, we vary both  the shear rate and applied pressure while keeping 
the inertial number $I = \dot \gamma \sqrt{m/P}$ constant.   The question is whether we 
can understand the memory loss and interaction network evolution in this simple setup.  

\begin{figure}
\includegraphics{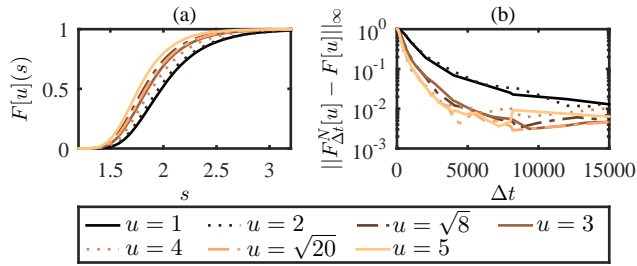}
\caption{(a) Estimated limiting distributions $F[u]$, given by $F^N_{2^{15}}[u]$, for different systems.  (b) Value of  $|| F^N_{\Delta t}[u] - F[u]||_\infty$ as a function of $\Delta t$, where $\Delta t = 1$ corresponds to $2\tau_c$.}
\label{fig::lim}
\end{figure}

There are three relevant time scales in the setup considered: $t_s ={1/{\dot \gamma}}$ (shear time), $t_I= \sqrt{m/P}$ (inertial time) 
and $t_c = \sqrt{m/k_n}$ (binary collision time, also comparable to the time needed for a signal to propagate through a grain).  
One could also think of the typical contact time between the particles as a relevant time scale, however, for a simple shear flow considered here, this time is comparable to $t_I$.  Out of these three time scales, one can produce two independent 
parameters; one possibility is to define $I = {t_I/t_s}$ and $\kappa = {t_I/t_c}$.    Since $t_c \ll t_I \ll t_s  $, $I \ll 1$ and $\kappa \gg 1$
(for the reference case, $I \approx 10^{-4}$ and $\kappa \approx 10^{5}$, see Section~\ref{sec:simulations}.   For such large values of $\kappa$, it has been 
argued that the main features of the flow are $\kappa$-independent, and that the effect of finite elastic modulus of the particles 
can be ignored~\cite{cruz_2005}.   Note that two simulations with different values of $P$ and $\dot \gamma, $ but with $I$ and 
$\kappa$ the same are identical after rescaling the time with  $t_s$.

The reference case of our simulations corresponds  to $P=1$ and  $v = v_{ref} \approx 2.5 \times 10^{-5}$. We refer to the system sheared with  $v = uv_{ref}$ by $S_u$.  For different values of $u$, we record the positions of the particles and forces acting between them at times  $t_i = 2 i\tau_c $, where $i = 0, \ldots, N = 5\times10^{6}$. By applying persistent homology to the force networks recorded at times $t_i$, we obtain a  sequence $\setof{x_i}_{i=0}^{N}$ that captures the evolution of the topological properties of the interaction network along the sampled trajectory.  To assess the differences between the persistence diagrams, we use the Wasserstein $d_{W^2}$ distance which mitigates the influence of noise~\cite{physicaD14}.

\section{ Memory loss of linearly sheared systems.} We use Theorem~\ref{thm::convergence_rate} to infer  the rate of memory loss for the systems $S_u$. Using the sample  $\setof{x^u_i}_{i=0}^{N}$ of the system $S_u$, we compute  the CDFs $F^N_{m}[u]$ of  $\setof{d_{W^2}(x^u_n,x^u_{n+m})}_{i = 0}^{N-m}$.  Figure~\ref{fig::Dist} suggests that $F^N_{m}[1]$ and $F^N_{m}[5]$ converge.  For every $S_u$ we verified that the functions $F^N_m[u]$  do not depend on the choice of initial conditions and  change very little if $m \geq 10^4$. Hence, we  approximate the limiting distributions $F[u]$ by  $F^N_{2^{15}}[u]$. As documented by Fig.~\ref{fig::lim}(a), the limiting distributions $F[u]$ vary with $u$ and there is no obvious trend. Thus, the asymptotic dynamics of the systems $S_u$ changes in a rather complicated manner.

To study the memory loss we  denote a CDF of  $\setof{d_{W^2}(x^u(t),x^u(t+{\Delta t}))}$ by  $F^N_{\Delta t}[u]$.  Figure~\ref{fig::lim}(b) shows  $|| F^N_{\Delta t}[u] - F[u]||_\infty$ as a function of $\Delta t$ for different systems $S_u$. It indicates that the rate of memory loss increases with $u$, but in a non-smooth fashion.   This suggests a presence of bifurcations along the parameter $u$. The rate of memory loss increases with $u$ as we pass over the consecutive bifurcation values while it stays essentially constant away from this values.

Based on the earlier discussion about expected minor influence of the parameter $\kappa$ measuring elasticity of particles, it is 
natural to analyze memory loss by considering the time scaled by $t_s$. To be specific, we rescale time by the wall speed in such a way that the wall moves by one particle diameter during one time unit. Using rescaled time $t_u = t  /(u v_{ref})$,  we consider $|| F^N_{\Delta t_u} [u]- F[u]||_\infty$ as functions of $\Delta t_u$.  They indicate that the decay rate of $|| F^N_{\Delta t_u} [u]- F[u]||_\infty$ varies erratically with $u$.

\begin{figure}
\includegraphics{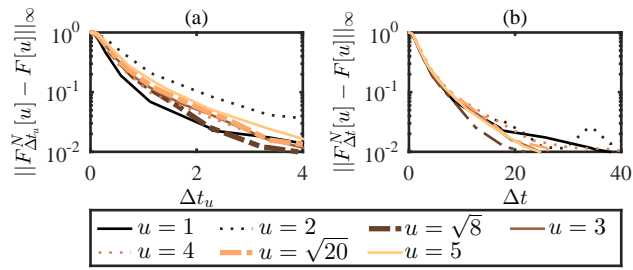}
\caption{(a) Value of  $|| F^N_{\Delta t_u} - F||_\infty$ as a function of $\Delta t_u$ for different systems. Time $t_u$ is scaled by the wall speed. (b) Value of  $|| F^N_{\Delta t^*_u} - F||_\infty$ as a function of $\Delta t^*_u$ for different systems. Time $t^*_u$ is scaled by the number of transitions.}
\label{fig::KS}
\end{figure}

This finding is surprising since, as mentioned earlier in the text, one would expect essentially identical behavior of the considered systems since $I$ is kept fixed, and the values of $\kappa$ are large.  To better understand this behavior, we recall that at the physical time scale the rate of memory loss does not change much with $u$ unless the parameter $u$ passes through an expected bifurcation value. At those values the rate increases with $u$, see Fig.~\ref{fig::lim}(b). On the other hand, Fig.~\ref{fig::KS}(a), shows that at the time scale $t_u$ the rate of memory loss decreases with $u$ between the consecutive bifurcations which is consistent with~\cite{cruz_2005}. However, the rate of memory loss still increases with $u$ as it crosses over a bifurcation value.

\begin{figure}
\includegraphics{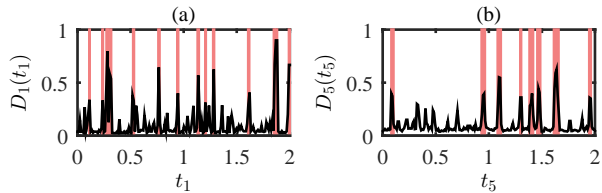}
\caption{ Values of  (a) $D_1(t_1)$ and (b) $D_5(t_5)$. The detected transitions are marked by red boxes.}
\label{fig::Transitions}
\end{figure}

Finally, we connect the unexpected behavior shown in Fig.~\ref{fig::KS}(a) directly to the 
evolution of interaction network.  The animations, indicate that this evolution exhibits slow-fast dynamics: slow dynamics is dominated by build up of force `chains' that buckle and lead to large and fast rearrangements.   We proceed to analyze these transitions and their possible connection to the memory loss.

\begin{figure}[t]
\includegraphics{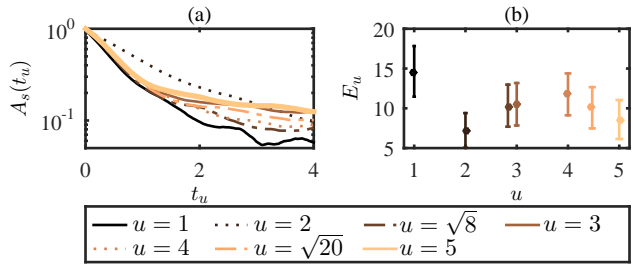}
\caption{(a)  Autocorrelation of the stress tensor. (b) Average number of transitions $E_u$ of the system $S_u$  during the time interval in which its wall moves by two particle diameters. The error bars indicate the standard deviation. }
\label{fig::DecayCompare}
\end{figure}

To quantify the transitions, we consider the evolution of the systems at the fastest timescale given by our sampling. Let  $D_u(t_u)$ be the $d_{W^2}$ distance between the state of the system $S_u$ sampled at time $t_u$ and the next sampled state. The peaks of $D_u$, visible in Fig.~\ref{fig::Transitions},  indicate presence of abrupt transitions of the system which otherwise exhibits a relatively slow evolution.  We define individual abrupt transitions as distinct peaks of $D_u$ that exceed a specified 
threshold; the number of detected transitions depends on the choice of threshold, however we find that 
the following results suggesting universal dependence of the memory loss on the number of abrupt transitions hold
for thresholds  in $ [0.2, 0.4]$. Note that if $\bar{D_u}$ is the mean of $D_u$ and  $\sigma_{D_u}$ is its standard deviation,  then  the interval  $[\bar{D_u} + \sigma_{D_u} , \bar{D_u} + 3\sigma_{D_u}]  \subset  [0.2, 0.4]$ for each $D_u$. We
use the threshold value of $0.3$.

Let $\delta t_{u}$ be an average time between two consecutive abrupt transitions of the system $S_u$ and consider the rescaled time  $ t^* =  t_u/\delta t_{u}$. Note that  $t^*$ does not depend explicitly on $u$, but only on the number  of transitions.  Figure~\ref{fig::KS}(b) shows that values of $||F^N_{\Delta t^*}[u] - F[u]||_\infty$ are very similar for all systems $S_u$.  This suggests that the rate of memory loss depends universally on the frequency of the transitions. Thus, we conjecture that the memory is predominately lost during the abrupt transitions caused by localized buckling of the interaction networks.

Now, we investigate influence of the abrupt transitions on global measures such as the stress tensor. Figure~\ref{fig::DecayCompare}(a) shows the temporal autocorrelation of the shear stress.  We see that the 
ordering of the autocorrelations for various values of $u$ is similar to the ordering of the curves $|| F^N_{\Delta t_u} - F||_\infty$, see Fig.~\ref{fig::KS}(a). Moreover, only a handful of transitions occur before the autocorrelation drops below $0.2$. Figure~\ref{fig::DecayCompare}(b) shows the  average number $E_u$ of transitions of the system $S_u$ during the time  that the wall moves by two particle diameters. Figures~\ref{fig::DecayCompare}(a)-(b) indicate that the systems with  large $E_u$ decorrelate faster. 

\section{ Conclusion.}
In summary, we formulated a method for estimating the rate of memory loss in complex spatio-temporal
systems. We employed it to study the time evolution of force networks of granular systems in the space of persistence diagrams. We showed that dependence of the dynamics on the control parameters is very complex even for the simple case of uniform flow in a planar geometry. The rate of the memory loss depends erratically on the stiffness of the particles even if the particles are stiff and the inertial number is fixed. While proper understanding of this dependence should be the subject of future research, we found that the rate of the memory loss is strongly correlated with the frequency of abrupt transitions of interaction network. The fact that even a simple planar flow of circular particles exhibits an extremely complex dependence on the control parameters suggests that significant new research is still needed in the field of dense granular matter. In particular, much more work will be needed to fully understand  the systems exposed to more complex flows or built from more complex particles.

{\bf Acknowledgements }
M.K. was supported by ERC Gudhi (ERC-2013-ADG-339025).  L.K. and L.K. were partially supported by NSF Grant No. DMS-1521717, DARPA contract HR0011-15-2-0033, and ARO Grant No.  W911NF1810184. 
K.M. was supported by NSF-DMS-1248071, DMS-1521771, DMS-1622401, DMS-1839294 and DARPA contracts HR0011-16-2-0033 and FA8750-17-C-0054.

\section{Appendix}

\subsection{Particle interactions and simulation protocol}
\label{sec:simulations}
The particles in the considered numerical system are modeled as 2D soft frictional inelastic disks that interact via normal and tangential forces, specified
here in nondimensional form. We use the average particle diameter, {\bf $d_{ave}$}, as
the lengthscale, the binary particle collision time  $\tau_c = 2\pi
\sqrt{d_{ave}/(2 g k_n)}$ as the time scale, and the average particle mass,
$m$, as the mass scale.  Force constant, $k_n$ (in units
of ${ m g/d_{ave}}$), is set to a value
corresponding to photoelastic disks~\cite{geng_physicad03}.  The
parameters entering the linear force model can be connected to
physical properties (Young modulus, Poisson ratio) as described
 e.g. in \cite{kondic_99}.

Dimensionless normal force between $i$--th and $j$--th particle is
\begin{equation}
\boldsymbol F_{i,j}^{n}=k_n x_{i,j}\boldsymbol n-\eta_n\overline m\boldsymbol v_{i,j}^{n}\,
\end{equation}
where $\boldsymbol v_{i,j}^{n}$ is the relative normal velocity, $\overline{m}$ is reduced mass, $x_{i,j} = d_{ave} - r_{i,j}$ is the amount of compression, with $d_{ave} = {(d_i + d_j)/2}$ and $d_i$, $d_j$ diameters of the particles $i$ and $j$. The distance of the centers of $i$--th and $j$--th particle is denoted as $r_{i,j}$. Parameter $\eta_n$ is the damping coefficient in the normal direction, related to the coefficient of restitution set to $e=0.5$.

We implement the Cundall--Strack model for static friction \cite{cundall79}. The tangential spring $\boldsymbol\xi$ is introduced between particles for each new contact that forms at time $T = T_0$ and is used to determine the tangential force during the contact of particles. Due to the relative motion of particles, the spring length $\xi$ evolves as $\xi=\int_{T_0}^T\boldsymbol v_{i,j}^{t}(t)dt$ with $\boldsymbol v_{i,j}^{t}=\boldsymbol v_{i,j}-\boldsymbol v_{i,j}^{n}$ and $\boldsymbol v_{i,j}$ being the relative velocity of particles $i,~j$. The tangential direction is defined as $\boldsymbol{t}={\boldsymbol v_{i,j}^{t}/|\boldsymbol v_{i,j}^{t}|}$. The direction of $\boldsymbol\xi$ evolves over time and we thus correct the tangential spring as   $\boldsymbol{\xi}^{}=\boldsymbol{\xi}-\boldsymbol{n}(\boldsymbol{n.\xi})$. The tangential force is set to
\begin{equation}
\boldsymbol F^{t} = \min(\mu |\boldsymbol F^{n} |, |\boldsymbol F^{t\ast} |){\boldsymbol F^{t\ast} /|\boldsymbol F^{t\ast} | }\, ,
\end{equation}
with
\begin{equation}
\boldsymbol F^{t\ast} = -k_t \boldsymbol \xi^{} - \eta_t  m\boldsymbol v_{i,j}^{t}\, .
\end{equation}
Viscous damping in the tangential direction is included in the model via the damping coefficient $\eta_t = \eta_n$.
The value of the normal spring constant is $k_n=4\times 10^3$ and parameters $\eta_n$  and $k_t$ are set to $\eta_n = 1.4$ (consistent with the specified value of $e=0.5$)  and $k_t=0.8k_n$. Friction coefficient is set to $\mu=0.7$.

In the simulations the particles are  placed on a base that moves with the speed that 
varies linearly from $0$ to the speed of the top wall.  The purpose of the base is to ensure linear velocity profile across the domain.  There is a dissipative effect from friction between the particles and the base; for the $i$-th particle
\begin{equation}
 - \mu_b|\boldsymbol g|{\boldsymbol a_i(t)\over |\boldsymbol a_i(t)|}
\end{equation}
where $\mu_b$ is the friction between particle and base and $\boldsymbol a_i$ is the acceleration of the $i$-th particle. The magnitude and direction of $\boldsymbol a_i$ is determined here from the interaction of the particle $i$ with all particles in contact
\begin{equation}
 \boldsymbol a_i = {1\over m_i}\sum_{c_i} \boldsymbol F_{i,c_i}
\end{equation}
where $c_i$ runs over all particles in contact with particle $i$.

In our simulations we integrate Newton's equations of motion for both the translational and rotational degrees of freedom using a $4$th order predictor-corrector method with time step $\Delta t =0.02$. Initially, particles are placed on a grid and given random initial velocity; particles are bidisperse with the ratio of the large to small particle diameter $1.4$. Approximately $1/3$ of the particles have large and $2/3$ of the particles have small diameter. There is $\approx 1200$ particles and the rectangular domain is $54$ particle diameters wide and $23$ particle diameters high (in terms of $d_{ave}$).

Our simulations start  by slowly compressing the domain with a specified pressure, $P = \rm k^2P_0$, applied on the top wall, until the top wall reaches a steady position. The system is then sheared with velocity $v =  kv_0$, where $v_{0}= 2.5\cdot 10^{-5}$ (in the units of $d_{ave}/\tau_c$). Walls are built of monodisperse particles with diameters of size $d_{ave}$ placed initially at equal distances, $d_{ave}$, from each other; in the horizontal direction the boundary conditions are periodic. The value of $P_0$ is found by compression of the top wall up to a packing fraction $\rho=0.80$ ensuring a dense packing above jamming point. We show results for different values of $k\in[1.0,3.0,4.0,\sqrt{20},5]$; note that the ratio of the applied pressure and shearing velocity guarantees fixed inertial number (as discussed in the main body of the paper). Specifically, the shear rate, $\dot{\gamma} = H/v$,
where $H$ denotes the height of the granular system and inertial number $I=\dot{\gamma}\sqrt{m/P}\approx 10^{-4}$.

\subsection{Forces and Effective Friction}

Interaction networks of the granular systems with the fixed inertial number, $I$, are expected to evolve in the same manner. In
particular, the velocity profile during shear should stay the same regardless of the pressure or shearing velocity applied and the measures such as probability distribution function of forces, $\rm PDF$, and effective friction, $\mu^*$ found as a ratio of shear and normal stress, attain the same value~\cite{cruz_2005}.

We show that, indeed, these global measures confirm the previous findings and that our conclusions do not depend on
the subtle differences such as the existence of the frictional base. Specifically, the ratio of the typical interparticle force, $\langle F \rangle$ and the force between the particle and frictional base, $F_b$, is $\langle F \rangle/F_b < 10^{-2}$. We also find that the ratio of the kinetic and elastic energy in each
system considered is $ \ll 1$, showing that the force dynamics is not influenced by the existence of the frictional base in the dense flow regime, characterized by $I\leq O(10^{-3})$~\cite{cruz_2005}. Figure~\ref{fig:pdf} shows the probability distribution functions of forces, $\rm PDF$, averaged over complete simulation. As expected in the case of a fixed inertial number, the $\rm PDF$ curves collapse for all the systems considered.

\begin{figure}[ht!]
  \centering
  \includegraphics{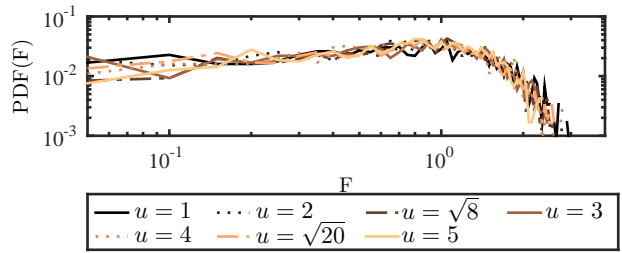}
  \caption{Probability distribution function of forces, $\rm PDF$, for all systems considered.}\label{fig:pdf}
\end{figure}

Next we discuss the effective friction, $\mu^*$. It has been shown~\cite{cruz_2005} that the value of $\mu^* < \mu$ for dense flows and that it remains constant for a large range of inertial numbers, $I\leq (10^{-3})$. Since in our case $I\approx 10^{-4}$ we expect to see similar results (i.e. $\mu^* < \mu$). Figure~\ref{effective_friction} shows $\mu^*$ for different systems as a function of the wall displacement, $d$ (in the units of $d_{ave}$). We observe $\mu^* \approx 0.45$ which is smaller than the interparticle friction $\mu=0.7$. Moreover, the value of $\mu^*$ fluctuates around the same value regardless of the pressure applied.

\begin{figure}[ht!]
 \centering
 \includegraphics{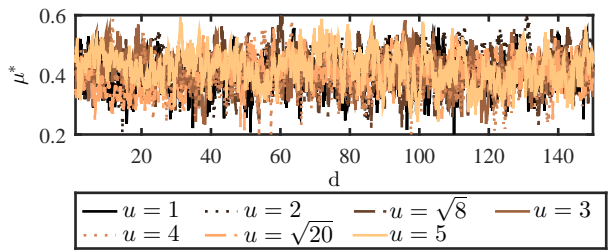}
 \caption{Effective friction, $\mu^*$ for different systems with fixed inertial number.}\label{effective_friction}
\end{figure}

We point out the fact that all the considered measures yield the same results for all the systems that we consider. In other words, the conclusions found in~\cite{cruz_2005} are confirmed here as well. However, our findings strongly suggest that not all essential information can be gathered by means of global measures and one needs to consider the topology of interaction networks as well.

\subsection{Strongly mixing systems}
\label{sec:strongly_mixing_systems}

In this section we consider a dynamical system  generated by  $f \colon X \to X$, where $X \subset \R^M$ is compact. We suppose that the system has a unique invariant measure $\mu$, and its dynamics is strongly mixing with respect to this  measure. We also require that the measure $\mu$ has a continuous Radon-Nikodym derivative with respect to the Lebesgues measure $dx$, i.e. there exists a continuous function $g \colon X \to X$ such  $\mu(A) = \int_A g dx$  for every measurable set $A \subseteq X$. Under these assumptions we get the following upper bound on the measure of the intersection of an $\varepsilon$ neighborhood of a sphere and subsets of $X$. 

\begin{lemma} Let $X\subset \R^M$ be a compact set and $\bar{\varepsilon}>0$.  For $x\in X$, $s\in \R$ and $\varepsilon >0$ we define 
\begin{equation}
\label{def::annulus}
\mathcal{A}(x, s,\varepsilon) := \setof{y \in X \colon s-\varepsilon \leq d(x,y) < s + \varepsilon},
\end{equation}
where $d$ is the Euclidean distance. If $\mu$ is a measure on $X$ whose Radon-Nikodym derivative is continuous  with respect to the Lebesgues measure on $\R^M$, then there exists a constant $K$ such that 
\begin{equation}
\label{def::annulus_measure}
\mu(\mathcal{A}( x, s,\varepsilon)) \leq K \varepsilon.
\end{equation}
for every $\varepsilon < \bar{\varepsilon}$

\label{lem::annulus_measure}
\end{lemma}

\begin{proof}
The compact set $X$ is contained in some ball with a sufficiently large radius $R$. Thus, volume of every set  $\mathcal{A}( x, s,\varepsilon)$ is smaller than the volume of an $\varepsilon$ neighborhood of an $M$-dimensional sphere with radius $R$ which is given by
\[
V(R,\varepsilon) = K_M((R + \varepsilon)^M - (R - \varepsilon)^M) = 2K_M\varepsilon\left( {N \choose 1}R^{N-1} +  {N \choose 3}R^{N-3}\varepsilon^2 + \ldots + \varepsilon^{N-1} \right),
\]
where $K_m = \frac{\pi^{\frac{M}{2}}}{\Gamma(\frac{M}{2}+1)}$. This implies that, for $\varepsilon < \bar{\varepsilon}$, the Lebesgues measure  of the set $\mathcal{A}( x, s,\varepsilon)$ is bounded by $K'\varepsilon$ where  $K' = 2K_M\left( {N \choose 1}R^{N-1} +  {N \choose 3}R^{N-3}\bar{\varepsilon}^2 + \ldots + \bar{\varepsilon}^{N-1} \right)$. Finally, by  continuity of the Radon-Nikodym derivative $g$ of the measure $\mu$ we get that $\mu(\mathcal{A}( x, s,\varepsilon)) \leq K\varepsilon$ for $K = K'\max_{x\in X}g(x)$.
\end{proof}

We use this lemma to prove convergence of CDFs, $F_m$, of the random variable $d(x,f^m(x))$ with $x$ distributed according to $\mu$. The main theorem of our Letter (restated below) guarantees that the mixing rate of the system provides an upper bound on $|| F_m - F ||_\infty$ where $F$ is a cumulative distribution of the random variable  $d(x,y)$ with $x,y$ independently distributed  according to $\mu$. Thus, the theorem implies that for strongly mixing systems  $|| F_m - F ||_\infty \to 0$ as $m \to \infty$.

\begin{theorem}
\label{thm::convergence_rate}
 Let $X\subset \R^M$ be compact and suppose that $f\colon X\to X$ has a unique invariant measure, $\mu$, whose Radon-Nikodym derivative is continuous  with respect to the Lebesgues measure on $\R^M$.  Let $\setof{\varepsilon_m}_{m=1}^{\infty}$ be a sequence of positive numbers converging to zero. If there exits a sequence of partitions $\setof{\mathcal{T}_m}_{m=1}^{\infty}$ of $X$ such that diameter of every set $T \in \mathcal{T}_m$ is less than $\varepsilon_m$ and for every measurable set $E \subset X$
\begin{equation}
|\mu(E)\mu(T) - \mu(f^{-m}(E) \cap T) | <  \varepsilon_m \mu(T),
\label{eqn::mixing}
\end{equation}
then there exits a constant $C > 0$ such that
\begin{equation}
\label{eqn::convergence}
|| F_m - F ||_\infty < C\varepsilon_m.
\end{equation}
\end{theorem}

\begin{proof}
It follows from additive property of the measure that the CDFs  $F_m$ and $F$ can be calculated as:
\begin{equation}
\label{eqn::Fn}
F_m(s) =  \sum_{T\in \mathcal{T}_m} \mu(\setof{x \in T \colon  d(x,f^m(x )) < s  })
\end{equation}
and
\begin{equation}
\label{eqn::F}
F(s) =  \sum_{T\in \mathcal{T}_m} \mu \times \mu (\setof{(x,y)\in T\times X \colon  d(x,y ) < s  }),
\end{equation}
where $ \mu \times \mu$ is the product measure. For every set $T \in \mathcal{T}_m$ we choose a random point $x_T \in T$ and define $B(x_T,s) := \setof{x \in X \colon d(x_T,x) < s}$. The following inequality follows from Equations~(\ref{eqn::Fn})~,~(\ref{eqn::F}) and the triangle  inequality:
\begin{align}
|F_m(s) - F(s) | & \leq  \sum_{T\in \mathcal{T}_m} \left| \mu(\setof{x\in T \colon  d(x,f^m(x )) < s }) -\mu(\setof{x\in T \colon  d(x_T,f^m(x ) ) < s  }) \right|  +  \nonumber \\
  & +   \sum_{T\in \mathcal{T}_m} \left| \mu(\setof{x\in T \colon  d(x_T,f^m(x )) < s  }) - \mu(B(x_T,s))\mu(T) \right|  \nonumber \\
  & +   \sum_{T\in \mathcal{T}_m} \left|  \mu(B(x_T,s))\mu(T)  - \mu \times \mu(\setof{(x,y)\in T \times X \colon  d(x,y) < s  }) \right|.  \nonumber \\
\label{eq:1}
\end{align}
To show that Inequality (\ref{eqn::convergence}) holds we need to properly bound the individual terms in the above inequality.

We start by estimating the terms in the first sum. 
We denote the symmetric difference of the sets $A$ and $B$ by $A\Delta B$ and recall that  $|\mu(A) - \mu(B)| \leq\mu(A\Delta B)$ for all measurable sets $A$ and $B$. Now, we estimate the symmetric difference of the sets
\begin{align}
 \setof{x\in T \colon  d(x,f^m(x )) < s } \Delta \setof{x\in T \colon  d(x_T,f^m(x ) ) < s  } &\subseteq  \setof{x\in T \colon f^m(x) \in \mathcal{A}(x_T,s, \varepsilon_m)} =  \nonumber \\
&= f^{-m}(\mathcal{A}(x_T,s,\varepsilon_m)) \cap T.  \nonumber 
\label{eq:1}
\end{align}
The inclusion follows from the fact that $B(x_T,s-\varepsilon_m ) \subset B(x,s) \subset B(x_T,s+\varepsilon_m)$ for all $x \in T$. It follows from  (\ref{eqn::mixing}) that
\[
| \mu((\mathcal{A}(x_T,s,\varepsilon_m))\mu(T) -\mu( T \cap f^{-m}(\mathcal{A}(x_T,s,\varepsilon_m)) )| \leq  \varepsilon_m\mu(T).
\]
By applying Lemma~\ref{lem::annulus_measure} we get
\[
\mu( T \cap f^{-m}(\mathcal{A}(x_T,s,\varepsilon_m)) ) \leq \varepsilon_m K\mu(T) + \varepsilon_m\mu(T).
\]
Therefore, the first sum in (\ref{eq:1}) is bounded by $\varepsilon_m(K+1)$.

Now, we turn our attention to the second sum . Note that 
\[
 \setof{x\in T \colon  d(x_T,f^m(x )) < s  } =  f^{-m}(B(x_T,s))\cap T.
\]
It follows from (\ref{eqn::mixing}) that $\left| \mu(f^{-m}(B(x_T,s))\cap T) - \mu(B(x_T,s))\mu(T) \right| < \varepsilon_m \mu(T)$
and the second sum is bounded by $\varepsilon_m$. Finally, we consider the last sum. By definition  $\mu(B(x_T,s))\mu(T) = \mu\times\mu( T \times B(x_T,s))$ and to bound the last sum by $K\varepsilon_m$ we just need to show that
\begin{equation}
    \left|  \mu\times\mu( T \times B(x_T,s))  - \mu \times \mu(\setof{(x,y)\in T \times X \colon  d(x,y) < s  }) \right| \leq K\varepsilon\mu(T)
    \label{eqn::last_sum}
\end{equation}
for every $T \in \mathcal{T}_m.$ By similar reasoning as above 
\[
\setof{T \times B(x_T,s)} \Delta\setof{(x,y)\in T \times X \colon  d(x,y) < s  } \subset T\times \mathcal{A}(x_T,s, \varepsilon_M)
\]
and  Inequality~(\ref{eqn::last_sum}) follows from the fact that $\mu(T\times \mathcal{A}(x_T,s, \varepsilon_M)) < \mu(T) K \varepsilon_m.$

By combining the estimates for individual sums we get that
\[
|F_m(s) - F(s) |  \leq  2(K+1)\varepsilon_m.
\]
for all $s \in \R$. Hence Inequality~(\ref{eqn::convergence}) holds for $C = 2(K+1)$.

\end{proof}

In practice we can only approximate the CDFs, $F_m$, from a finite sample $\setof{x_i}_{i=0}^N$  of a trajectory starting from some initial condition $x_0$. The approximation based on this sample is defined by
\[
F^N_m(s) = \frac{1}{N-m}\sum_{i=0}^{N-m} \chi_s( d(x_i,x_{i+m}) ).
\]
where $\chi_s$ is the characteristic function of the set $\setof{x \in \R \colon x < s}$.  We close this section by proving,  that if $f$ is a Borel measurable function, then $\lim_{M\to \infty} F_m^N(s) = F_m(s)$,  for every $s\in \R$ and $m \in \N$. By definition 
\[
F_m(s) :=  \int_{x\in X}  \chi_s(  d(x,f^m(x )) )d\mu.
\]
Hence to  prove that $F_m^N$ converges to  $F_m$  point-wise we just need to show that 
\[
\lim_{N \to \infty} \frac{1}{N-m}\sum_{i=0}^{N-m} \chi_s( d(x_i,x_{i+m}) ) = \int_{x\in X}  \chi_s(  d(x,f^m(x )) )d\mu,
\]
for $s \in \R$. This follows from Birkhoff ergodic theorem for almost every  $x_0 \in X$ under the assumption that $\chi_s(  d(x,f^m(x )) )$ is a Borel measurable function. If $f$ is  Borel measurable, then the map $x \to (x,f(x))$ is Borel measurable. Moreover, the distance function $d$ is Borel measurable with respect to the product algebra and so the functions $\chi_s$ are  Borel measurable  as well.

\subsection{Tent map} 
\label{sec:tent_map}

In this section we consider the dynamical system generated by the ten map
\[
f(x) := 2\min\setof{x, 1-x} \colon [0,1]\to [0,1]. 
\]
The invariant measure of this system coincides with the Lebesgues measure and we can use Theorem~\ref{thm::convergence_rate} to estimate the convergence rate of $||F_m - F||_\infty$. Lets us consider a family of partitions $\setof{\mathcal{T}_m}$  of $[0,1]$ where $\mathcal{T}_m$ consists of the intervals $T_m^i = [i2^{-m},(i+1)2^{-m}]$ for $i \in \setof{0, 1, \ldots 2^m-1}$. The radius of each interval $T_m^i$ is equal to $2^{-m}$. Moreover, every interval $T_m^i$ is uniformly starched to $[0,1]$ by $f^m$. This implies that for a fixed $m$ and every measurable  $E \subset [0,1]$ the measure of $\mu(f^{-m}(E) \cap T_m^i)$ is the same for all $i$. Because $f$ preserves $\mu$ it has to be equal to $\mu(E)2^{-m}$ which is exactly $\mu(E)\mu(T_i)$ and so $|\mu(f^{-m}(E) \cap T_m^i) -  \mu(E)\mu(T_i)| = 0$. Now it follows from Theorem~\ref{thm::convergence_rate} that 
\begin{equation}
\label{eqn::tent_convergence_rate}
||F_m - F||_\infty \leq C2^{-m}
\end{equation}
which is the inequality presented in our Letter.

The fact that the invariant measure $\mu$ coincides with Lebesgue measure on $[0,1]$ makes it easy to calculate the limiting distribution 
\[
F(s) =\left\{
     \begin{array}{ll}
        	0 \text{ if\;} s < 0 ,\\
         2s - s^2 \text{ if\;} 0 \leq s \leq 1 ,\\
    	 1 \text{ if\;}  s > 1.
     \end{array}
   \right.\\
\]
This allows us to study the convergences rate $||F_m^N - F||_\infty$ of the approximate distributions obtained from a finite sample $\setof{x_i}_{i = 0}^N$, $N = 5\times 10^6$,  of a trajectory  starting from  $x_0 =0.1$. Due to the binary floating-point representation of the numbers in a computer, the direct iteration of the tent map does not produce a good sample of the trajectory.  Instead we  use the fact that the tent map is homeomorhpic to the logistic map with parameter $r=4$. If we  denote the logistically evolving variable by $y_n$, then the iterations of the tent map are given by $x_{n}={\tfrac  {2}{\pi }}\sin ^{{-1}}(y_{{n}}^{{1/2}})$. We showed in this Letter that the bound on the convergence rate given by (\ref{eqn::tent_convergence_rate}) does not hold for $||F_m^N - F||_\infty$ if $m > 5$. This is due to the fact that the sample is not large enough and $F_m^N$ does not approximate $F_m$ with necessary  accuracy required to observe the predicted decay rate.

\subsection{Logistic map} 
\label{sec:logistic_map}

\begin{figure*}
\includegraphics{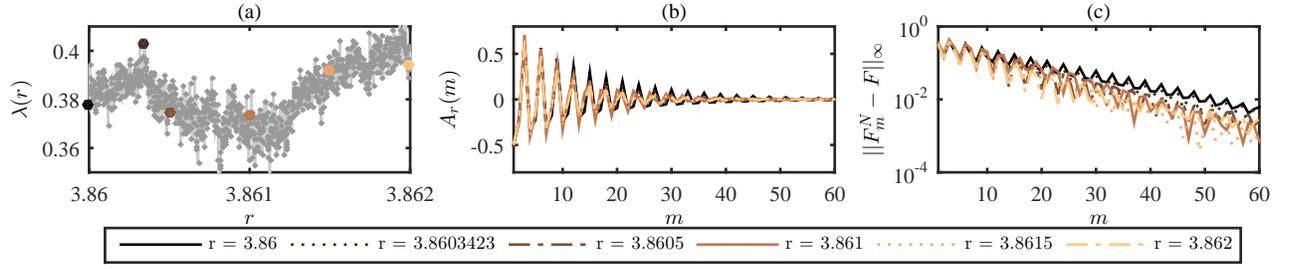}
\caption{(a) Lyapunov exponents $\lambda(r)$ for logistic map for different values the parameter $r$. Lyapunov exponents for the values of $r$ considered in (b) and (c) are emphasized. (b) Autocorrelation $A_r$ of the states for the logistic map for different values of $r$.  (c) Decay of $|| F^N_m - F ||_\infty$  is faster for the systems  exhibiting faster decay of correlations.} 
\label{fig::LogisticMap}
\end{figure*}

Our main theorem shows that convergence rate of $||F_m^M - F||_\infty$ can be used to infer the mixing rate of the system. The mixing rate is known to influence decay of correlations between the states of the system~\cite{wiggins2004foundations}.    For the logistic map we demonstrate that convergence rate of $||F_m^M - F||_\infty$ is closely related to the decay of correlations and  thus well suited for measuring the memory loss of the system.

The logistic map 
\[
f_r(x) := rx(1-x)  \colon [0,1] \to [0,1]
\]
is a classical example showing that complex chaotic behaviour can be produced by very simple non-linear systems. Dependence of the dynamics on the parameter $r$ is very complicated  and exhibits a large number of different bifurcations. This is documented by erratic behaviour of the Lyapunov exponents, shown in Fig.~\ref{fig::LogisticMap}(a), which measure the complexity of the dynamics. However, this measure is rather local and is not directly connected to the mixing rate.

Figure~\ref{fig::LogisticMap}(b) depicts the decay of correlations for different values of $r$. Generally, the decay rate increases with $r$ and does not reflect the erratic behavior of the Lyapunov exponents. The decay of $||F_m^M - F||_\infty$ is shown in Fig.~\ref{fig::LogisticMap}(c). As expected the systems with faster decay of correlations are mixing faster and thus $||F_m^M - F||_\infty$ also decays faster for these systems. This suggest that our measure  is indeed well suited for assessing the rate of memory loss.

\end{document}